\documentclass[11pt]{llncs}
\usepackage{amsmath,amssymb,amsfonts}
\usepackage{gastex}

\title{On the Synchronization Rate for e-machines}
\author{Mikhail V. Berlinkov\thanks{Supported by research grant of Prof. J\"{u}rgensen, University of Western Ontario, Canada}}

\authorrunning{Mikhail V. Berlinkov}

\tocauthor{Mikhail V. Berlinkov (Ekaterinburg, Russia)}

\institute{ Institute of Mathematics and Computer Science,\\
Ural Federal University, 620000 Ekaterinburg, Russia\\
\email{berlm@mail.ru}}

\DeclareSymbolFont{rsfscript}{OMS}{rsfs}{m}{n}
\DeclareSymbolFontAlphabet{\mathrsfs}{rsfscript}

\begin{document}
\maketitle

\begin{abstract}
It is known, that an $\epsilon$-machine is either exactly or
asymptotically synchronizing. In the exact case, the observer can
infer the current machine state after observing $L$ generated
symbols with probability $1-a^L$ where $0 \leq a<1$ is a so-called
synchronization rate constant. In the asymptotic case, the
probability of the correct prediction the current machine state
after observing $L$ generated symbols tends to $1$ exponentially
fast as $1-b^L$ for $0<b<1$ and the infimum of such $b$ is a
so-called prediction rate constant.

Hence the synchronization and prediction rate constants serve as
natural measures of synchronization for $\epsilon$-machines. In the
present work we show how to approximate these constants in
polynomial time in terms of the number of machine states.
\end{abstract}

\section{Preliminaries}

A \emph{deterministic finite automata} (DFA) $\mathrsfs{A}$ is a
triple $\langle Q,\Sigma,\delta \rangle$ where $Q$ is the state set,
$\Sigma$ is the input alphabet and $\delta: Q \times \Sigma
\rightarrow Q$ is the \emph{transition function}. If $\delta$ is
completely defined on $Q \times \Sigma$ then $\mathrsfs{A}$ is
called \emph{complete}, otherwise $\mathrsfs{A}$ is called
\emph{partial}. The function $\delta$ extends uniquely to a function
$Q\times\Sigma^*\to Q$, where $\Sigma^*$ stands for the free monoid
over $\Sigma$; the latter function is still denoted by $\delta$.
When we have specified a DFA $\mathrsfs{A}=\langle
Q,\Sigma,\delta\rangle$, we can simplify the notation by writing
$S.w$ instead of $\{ \delta(q,w) \mid q \in S \}$ for a subset
$S\subseteq Q$ and a word $w \in \Sigma^*$. In what follows, we
assume $|\Sigma|>1$ because the singleton alphabet case is trivial
for considered problems. Given a subset of words $U \subseteq
\Sigma^{*}$ and a state $p$ denote by $U_p$ the subset of words from
$U$ labeling some path from $p$ in $\mathrsfs{A}$.

A DFA $\mathrsfs{A} = \langle Q,\Sigma,\delta \rangle$ is called
\emph{synchronizing} if there exists a word $w\in\Sigma^*$ such that
$|Q.w|=1$. Notice that here $w$ is not assumed to be defined at all
states. Each word $w$ with this property is said to be a
\emph{reset} word for $\mathrsfs{A}$. The minimum length of such
words is called the \emph{reset threshold} of $\mathrsfs{A}$ and is
denoted by $rt(\mathrsfs{A})$.

The synchronization of strongly connected partial automata as models
of \emph{$\epsilon$-machines} is one of the central object for
research in the theory of stationary information sources. The
synchronization and state prediction for stationary information
sources has many applications in information theory and dynamical
systems. An $\epsilon$-machine can be defined as a strongly
connected DFA with a probability distribution defined on the
outgoing arrows for each state, without states with
\emph{probabilistically equivalent future}
(see~\cite{Emach},\cite{Amach} for details). Hence by
$\epsilon$-machine we mean the tuple $\langle
Q,\Sigma,\delta,\textbf{P} \rangle$ where $\langle Q,\Sigma,\delta
\rangle$ is the strongly connected partial automaton and
$\textbf{P}: Q \times \Sigma \mapsto R_{+}$\footnote{We write
$\textbf{P}_q(a)$ for the probability of generating $a$ from $q$
instead of $\textbf{P}(q,a)$} is the probability distribution on the
outgoing arrows. Formally,
\begin{itemize}
\item{$\sum_{a \in \Sigma}\textbf{P}_q(a) = 1$ for each state $q \in Q$;}
\item{$\textbf{P}_q(a) = 0$ whenever $a$ is undefined on $q \in Q$;}
\item{for each states $p \neq q$ there is a word $u$ such
that $\textbf{P}_p(u) \neq \textbf{P}_q(u)$,}
\end{itemize}
where the probability $\textbf{P}_p(v)$ of generating a given word
$v = au$ from the state $p$ is defined inductively by
$\textbf{P}_p(a)\textbf{P}_{p.a}(u)$.

An $\epsilon$-machine is \emph{exactly synchronizable} or simply
\emph{exact} if the corresponding partial strongly connected
automaton is synchronizing in our terms. Given an initial
probability distribution $\pi: Q \mapsto R_{+}$ on the states of an
$\epsilon$-machine, the probability of generating a word $v$ is
equal to the weighted sum $\textbf{P}_{\pi}(v) = \sum_{q \in
Q}{\pi_q \textbf{P}_q(v)}$.

Based on various applications, there are two basic settings for
synchronization. In the first setting, the observer knows that an
error appeared in the corresponding system and he can apply the
reset sequence for reestablishing correct behavior of the system.
Apparently, the most natural measure of synchronization in this
setting is the reset threshold of the corresponding automaton. This
is because the reset threshold corresponds to the minimum time
required to find out the current system state.

In the second setting, observer either doesn't know that an error
appeared or cannot affect to the system after that. For such
situations it is natural to estimate either the probability of
generating a reset word (if it exists) by the system itself or the
average uncertainty in predicting the current state by the generated
sequence in a given period of time.

\section{Computing Synchronization Rate}

Let $\mathrsfs{A} = \langle Q,\Sigma,\delta,\textbf{P} \rangle$ be
an exact $\epsilon$-machine and denote $n=|Q|$. It has been proved
in~\cite{Emach} that the probability of generating only non-reset
words of length $L$ decay exponentially fast as
$src(\mathrsfs{A})^L$ where $0<src(\mathrsfs{A})<1$ is a so-called
\emph{synchronization rate constant} can be defined as
$$src(\mathrsfs{A}) = \lim_{L \rightarrow +\infty}{(\textbf{P}_{\pi}(NSYN_L))^{1/L}}$$ where
$NSYN_L$ is the set of non-reset words of length $L$ and $\pi$ is a
steady state distribution on the states. Since $\mathrsfs{A}$ is
strongly connected, this value doesn't depend on the initial
distribution $\pi$. It is also shown in~\cite{Emach} that the
synchronization rate constant can be approximated with any given
accuracy in exponential time in terms of the number of machine
states $n$. In this section we show how to approximate
$src(\mathrsfs{A})$ with any given accuracy in polynomial time. As
well as in~\cite{Emach}, for this purpose, let us consider the
auxiliary \emph{semi $\epsilon$-machine}\footnote{By semi
$\epsilon$-machine we mean a partial automaton with edges weighted
by real numbers in the range $[0,1]$. We still refer to these
weights as probabilities.} $\mathrsfs{A}_2 = \langle Q_2 , \Sigma,
\delta_2 ,\textbf{P} \rangle$ where $Q_2 = \{ (p,q) \mid p,q \in Q,
p \neq q \}$. Given a pair of different states $p,q \in Q$ and a
letter $x \in \Sigma$, define
\begin{equation}\delta_2((p,q),x) = \begin{cases}
(\delta(p,x),\delta(q,x)),\ & |\{p.x,q.x\}|=2 \\
\text{undefined otherwise}.
\end{cases}
\end{equation}
\begin{equation}
\label{Pr2def} \textbf{P}_{(p,q)}(x) = \begin{cases}
\textbf{P}_p(x),\ & |\{p.x,q.x\}|=2 \\
0,\ &  \text{ otherwise}.
\end{cases}
\end{equation}

Let $T(\mathrsfs{A}_2,x)$ be the $n(n-1) \times n(n-1)$ matrix of
the transition probabilities of $\mathrsfs{A}_2$ for the letter $x
\in \Sigma$ indexed by the states of $Q_2$. That is, for each two
states of $s,t \in Q_2$ the entry on the intersection of $s$-th row
and $t$-th column is given by
\begin{equation}T(\mathrsfs{A}_2,x)_{s,t}=
\begin{cases}
\textbf{P}_s(x),\ & \delta_2(s,x) = t \\
0,\ &  \text{ otherwise}.
\end{cases}
\end{equation}
The transition probability matrix $T(\mathrsfs{A}_2)$ of
$\mathrsfs{A}_2$ is now defined as the sum $$T(\mathrsfs{A}_2) =
\sum_{x \in \Sigma}{T(\mathrsfs{A}_2,x)}.$$

Let $R^L_{p,q}$ be the sum of $(p,q)$-th row entries of
$T(\mathrsfs{A}_2)^L$. Notice that by definition~\ref{Pr2def},
$R^L_{p,q}$ is the probability of generating words by
$\epsilon$-machine $\mathrsfs{A}$ of length $L$ from $p$ which
doesn't merge the pair $\{p,q\}$. Define also $$R^L_{p} = \sum_{q
\in Q, p \neq q}{R^L_{p,q}} \quad \text{ and } \quad MaxR^L_{p} =
\max_{q \in Q, p \neq q}{R^L_{p,q}}.$$ A construction of a
polynomial time algorithm for approximation synchronization rate
constants is based on the following theorem.
\begin{theorem}
\label{th_comp_synchconst}For each state $p \in Q$ the probability
of generating a non-reset sequence of length $L$ from $p$ is upper
bounded by $R^L_{p}$ and is lower bounded by $MaxR^L_{p}$.
\end{theorem}
\begin{proof}
Let $v$ be a non-reset word of length $L$ generated by
$\mathrsfs{A}$ from the state $p$. Then there exists another state
$q \in Q$ such that $v$ doesn't merge $p$ and $q$ in $\mathrsfs{A}$.
This means that $v$ takes the pair $(p,q)$ to another proper pair
$(p',q')$ and thus the probability of generating this word is
included in $R^L_{p,q}$. The upper bound follows.

Now let $q$ be the state such that $R^L_{p,q} = MaxR^L_{p}$. Let $u$
be a word of length $L$ which doesn't merge the pair $(p,q) \in
Q_2$, that is, $u$ is an arbitrary word which probability is
included in the $(p,q)$-th row. Then $u$ is non-reset and the lower
bound follows.
\end{proof}

Since $\mathrsfs{A}$ is strongly connected and synchronizing, it has
a unique positive steady state distribution $\pi \in R^n_{+}$
(see e.g.~\cite{MyMchains}). Hence we get the following corollary.
\begin{corollary}
\label{cor_bounds}
\begin{equation}
\label{eq_bounds} \sum_{p \in Q}{\pi_p MaxR^L_{p}} \leq
\textbf{P}_{\pi}(NSYN_L) \leq \sum_{p \in Q}{\pi_p
R^L_{p}}.
\end{equation}
\end{corollary}

Thus we get that the probability of generating non-reset words of
length $L$ is lower bounded by $\pi_{min}
\parallel T(\mathrsfs{A}_2)^L \parallel_{1}$ and is upper bounded by
$n^2 \pi_{max} \parallel T(\mathrsfs{A}_2)^L \parallel_{1}$.
Taking the power $1/L$ we get that the synchronization rate constant
is given by the maximal eigenvalue of the transition probabilitiy
matrix $T(\mathrsfs{A}_2)$ whence the main result of this section
follows.
\begin{theorem}
\label{main_th} Given an exact $n$-state $\epsilon$-machine
$\mathrsfs{A}$, its synchronization rate constant can be
approximated in a polynomial time $\phi(n,\delta)$ for arbitrary
small absolute error $\delta>0$.
\end{theorem}
Notice also, that Corollary~\ref{cor_bounds} can be used to estimate
the probability of generating non-reset words of a given length in a
polynomial time.

\section{Computing Prediction Rate}

In the previous section we show how to compute synchronization rate
constants for exact $\epsilon$-machines. It turns out that if an
$\epsilon$-machine is not exact, it is still can be synchronized but
only asymptotically, that is, for almost every infinite word
$\overrightarrow{x}$ the observer uncertainty decay exponentially
fast as $a^L$ after reading $L$ first symbols of $x$~\cite{Amach}.

Let an $\epsilon$-machine $\mathrsfs{A} = \langle Q=\{1,2, \dots,
n\},\Sigma,\delta,\textbf{P} \rangle$ generates a word $w$, denote
by $\phi(w) \in R^n$ the observer \emph{belief distribution}, that
is,
$$\phi(w)_q = \frac{\sum_{p \in Q \mid p.w = q}{\pi_p \textbf{P}_p(w)}}{\sum_{p \in Q}{\pi_p \textbf{P}_p(w)} }$$ where $q \in \{1,2,\dots, n\}$ and $\pi$ is the initial distribution of $\mathrsfs{A}$.
Clearly, $\phi(w)$ is a stochastic vector and $q$-th entry is equal
to the probability that $\mathrsfs{A}$ is in the state $q$ after
generating $w$. In particular, if $w$ is reset then $\phi(w)$ has
only one non-null entry, equals $1$. Let $\Phi_L \equiv \phi:
\Sigma^L \mapsto R^n$ be the random variable for the belief
distribution over states induced by the first length-$L$ word the
machine generates, and $\overline{S}_L$ be the most likely state in
$\Phi_L$ (if a tie the lowest numbered state is taken). Denote by
$Q_L \equiv 1-\textbf{P}(\overline{S}_L)$ the combined probability
of all other states in the distribution $\Phi_L$. Then
$\mathrsfs{A}$ is said to be \emph{asymptotically synchronizing} if
$\textbf{P}(Q_L > \delta)$ vanishes to $0$ while $L \rightarrow
+\infty$ for each $\delta>0$. In other words, for almost every
infinite word the observer uncertainty in predicting the current
machine state vanishes to $0$ (see~\cite{Amach} for details).

It is proved in~\cite{Amach} that actually each $\epsilon$-machine
is asymptotically synchronizing, and $Q_L$ vanishes to $0$
exponentially fast, that is, $\textbf{P}(Q_L > a^L)$ also vanishes
to $0$ for some $0<a<1$ while $L \rightarrow +\infty$. However, no
algorithm for computing infimum of such $a$ is given
in~\cite{Amach}. Let us call this infimum the \emph{prediction rate
constant} of an $\epsilon$-machine and denote it by
$prc(\mathrsfs{A})$. Below we show that the prediction rate constant
is always positive whenever $\epsilon$-machine is not exact, present
a polynomial time algorithm for approximating it, and simultaneously prove the
aforementioned result from~\cite{Amach}, apparently in a simpler
way. Notice also that the prediction rate constant is equal to $0$
for exact $\epsilon$-machines.

In what follows, let $\mathrsfs{A}$ be non-exact. Suppose
$\mathrsfs{A}$ generates a word $w$. Denote by $f_w$ the state with
maximal $\textbf{P}_{p}(w)$ among $p \in Q$ and by $s_w$ the state
with maximal $\textbf{P}_p(w)$ among $p \in Q$ such that $p.w \neq
f_w.w$ (if a tie the lowest numbered state is taken). In these
terms, we can bound $Q_L(w)$ as follows.

\begin{lemma}
\label{lem_ql_bounds} For some constants $c_1,c_2$ depended only on
the machine, we have
\begin{equation}
\label{eq_q_bounds}c_1
\frac{\textbf{P}_{s_w}(w)}{\textbf{P}_{f_w}(w)} \leq Q_L(w) \leq c_2
\frac{\textbf{P}_{s_w}(w)}{\textbf{P}_{f_w}(w)}.
\end{equation}
\end{lemma}
\begin{proof}
By the definition of $Q_L$ we have $Q_L(w) = \frac{\sum_{p \in Q
\mid p.w \neq q}{\pi_p \textbf{P}_p(w)}}{\sum_{p \in Q}{\pi_p
\textbf{P}_p(w)}}$ for $q \in Q$ with maximal $\phi(w)_q$.

Since $f_w.w \neq s_w.w$, either $f_w.w \neq q$ or $s_w.w \neq q$
whence at least one of the components $\pi_{s_w}\textbf{P}(w),
\pi_{f_w}\textbf{P}(w)$ is contained in the sum $\sum_{p \in Q \mid
p.w \neq q}{\pi_p \textbf{P}_p(w)}$. Since also $\textbf{P}_{s_w}(w)
\leq \textbf{P}_{f_w}(w)$ we get
$$\sum_{p \in Q \mid p.w \neq q}{\pi_p \textbf{P}_p(w)} \geq
\pi_{min}\textbf{P}_{s_w}(w),$$ where $\pi_{min} = \min_{q \in
Q}{\pi_q}$. By the choice of $q$ and $s_w$ we get
$$\sum_{p \in Q \mid p.w \neq q}{\pi_p \textbf{P}_p(w)} \leq 1-\pi_{f_w}\textbf{P}_{f_w}(w) \leq \pi_{s_w}\textbf{P}_{s_w}(w).$$
By the definition of $f_w$ we also get $$\pi_{min}
\textbf{P}_{f_w}(w) \leq \sum_{p \in Q}{\pi_p \textbf{P}_p(w)} \leq
\textbf{P}_{f_w}(w).$$ Thus we can choose $c_1 = \pi_{min}, c_2 =
\frac{\pi_{max}}{\pi_{min}}$ and the lemma follows.
\end{proof}

The following lemma gives us a way to calculate the prediction rate
constant.
\begin{lemma}
\label{lem_calc_PR}
$$\lim_{L \rightarrow \infty}{\frac{1}{\mathbb{E}(1/Q_L^{1/L})}} \leq prc(\mathrsfs{A}) \leq \lim_{L \rightarrow \infty}{\mathbb{E}(Q_L^{1/L})}.$$
\end{lemma}
\begin{proof}
Let $a>0, \delta>0$ be some constants. if $\textbf{P}(Q_L >
(a+\delta)^L)$ does not vanish to $0$, then $\mathbb{E}(Q_L^{1/L})
\geq a$ for $L$ big enough. From the other hand, if $\textbf{P}(Q_L
> (a-\delta)^L)$ vanishes to $0$ then $\mathbb{E}(1/Q_L^{1/L}) \geq
1/a$ for $L$ big enough. The lemma follows.
\end{proof}

Thus, in order to design an algorithm for computing prediction rate
constant, it is enough to calculate $\mathbb{E}(Q_L^{1/L})$ and then
to prove that $$1/\mathbb{E}(1/Q_L^{1/L}) \sim
\mathbb{E}(Q_L^{1/L}).$$

For these purposes we need the definition of associated
\emph{edge-machine} as in~\cite{Amach}. Let $M$ be an ergodic
irreducible Markov chain with the equilibrium distribution $(\rho_1,
\rho_2, \dots, \rho_l)$. Let $I(w, k, j)$ denotes the indicator
function of the transition from the state $k$ to the state $j$ by
the word $w$. That is, $I(w, k, j) = 1$ if $\delta(k,w)=j$ and $0$
otherwise. The edge machine $M_{edge}$ is the Markov chain whose
states are the outgoing edges of the original machine $M$. That is,
the states are the pairs $(x, k)$ such that $\textbf{P}_k(x) > 0$,
and the transition probabilities are defined as: $\textbf{P}((x, k)
\mapsto (y, j)) = \textbf{P}_j(y)I(x, k, j)$. A sequence of
$M_{edge}$ states visited by the Markov chain corresponds to a
sequence of edges visited by the original machine $M$ with the same
probabilities. It follows, that $M_{edge}$ is also ergodic and the
following remark is also straightforward.

\begin{remark}
\label{rem_edge_prob}The equilibrium distribution of $M_{edge}$ is
given by $(\rho_{p,x})_{p \in M, x \in \Sigma}$ where $\rho_{p,x} =
\rho_{p}\textbf{P}_p(x)$.
\end{remark}

In contrast to~\cite{Amach}, we consider $M_{edge}$ machines for
\emph{deadlock components} of original semi-machine $\mathrsfs{A}_2$
instead of $\mathrsfs{A}$. As in the previous section we consider
the same auxiliary automaton $\mathrsfs{A}_2 = \langle Q_2, \Sigma,
\delta_2 ,\textbf{P} \rangle$. Since $\mathrsfs{A}$ is not exact,
there are pairs $\{p,q\}$ that cannot be synchronized, that is
$|\{p,q\}.v| \neq 1$ for each word $u$ (this follows
from~\cite{Emach}[Theorem~3]). Such pairs are called
\emph{deadlock}. Let $\{p,q\}$ be a deadlock pair. It follows from
the definition, that for each $a \in \Sigma$ if $p.a$ is defined
then $q.a$ is also defined and $\{p.a,q.a\}$ is also deadlock. Hence
there are closed under the actions of the letters strongly connected
components of deadlock pairs in $\mathrsfs{A}_2$. Since
$\mathrsfs{A}$ is aperiodic, these components are also irreducible
$\epsilon$-machines. Let us denote the set of such components by
$\mathrsfs{M}$.

The following remark follows from the fact that the probability of
coming to closed components of the graph is positive.
\begin{remark}
\label{rem_out_of_M}There is a constant $0<c<1$ such that
$$\textbf{P}_z(\{w \mid \exists {p,q} (|\{p,q\}.w|>1, (p,q).w \notin
\mathrsfs{M} ) \} ) \leq c^{|w|}, \text{ for each } z \in Q.$$
\end{remark}
In simpler terms, it is stated in the remark that with probability
asymptotically closed to $1$ all pairs of states in $\mathrsfs{A}_2$
either merge or come to deadlock pairs. Let us also notice that the
constant $c$ from the remark is given by the dominant eigenvalue of
the transition matrix of $\mathrsfs{A}_2$ without deadlock
components whence it can be approximated in polynomial time.

As in~\cite{Amach} we use the theorem from~\cite{Glynn2002143} which
in simpler terms can be formulated as follows.

\begin{lemma}[\cite{Glynn2002143}]
\label{lem_mc_vars}Let $Z_0,Z_1, \dots$ be a finite-state,
irreducible Markov chain, with state set $R = \{1,2, \dots, K\}$ and
equilibrium distribution $\rho = (\rho_1, \rho_2, \dots, \rho_K)$.
Let $F: R \mapsto \mathbb{R}, Y_L = F(Z_L)$, and ${\overline{Y_L}} =
\frac{1}{L}(Y_0 + \dots + Y_{L-1})$. Denote also $E_{\rho}(F) =
\sum_{k}{\rho_k F(r_k)}$. Then, there exist $0<\alpha<1, \beta > 0$
such that for each $\epsilon>0$ and each state $k$ for $L$ big
enough we have
$$\textbf{P}_{k}( | \overline{Y_L} - E_{\rho}(F) | \geq \epsilon) \leq e^{- L \beta \epsilon^2} = \alpha^{L}.$$
\end{lemma}

Now let $M \in \mathrsfs{M}$ and $M_{edge}$ be the corresponding
edge-machine. Given a state $r$ of $M_{edge}$, or equivalently an
edge between two states $(p,q)$ and $(p',q')$ labeled by $x$, define
$F(r) = \ln{\frac{\textbf{P}_{p}(x)}{\textbf{P}_{q}(x)}}$. Due to
Lemma~\ref{lem_mc_vars} for any $\epsilon>0$, we have
$$\textbf{P}_{(p,q)}( | \overline{Y_L} - E_{\rho}(F) | > \epsilon ) \leq \alpha^L$$
for $L$ big enough.

Given any initial pair of states $(p,q)$ and a random word $w$ of
length $L$, we have $$[\overline{Y_L} \mid Z_0 = (p,q)](w) =
\frac{1}{L}\ln{\frac{\textbf{P}_{p}(w)}{\textbf{P}_{q}(w)}}$$ by the
definition of $M_{edge}$ and $F$. Hence we get
\begin{equation}
\label{eq_aep} \textbf{P}(\exp{(\mathbb{E}_{M}-\epsilon)} \leq
(\frac{\textbf{P}_p(w)}{\textbf{P}_q(w)})^{1/L} \leq
\exp{(\mathbb{E}_{M}+\epsilon)} ) \geq 1-\alpha^L,
\end{equation}
where $\mathbb{E}_{M} = E_{\rho}(F)$.

\begin{lemma}
\label{lem_about_avg}$\mathbb{E}_{M} > 0$ for each $M \in
\mathrsfs{M}$.
\end{lemma}
\begin{proof}
It is enough to prove that for each pair $(p,q)$ of $M$ the sum
$\sum_{x \in \Sigma_p }{\rho_{(p,q,x)} F((p,q,x))}$ is non-negative
and there is a pair for which this sum is positive. Since
$\rho_{(p,q,x)} = \rho_{(p,q)} \textbf{P}_p(x)$, by
Remark~\ref{rem_edge_prob} this sum equals
$$\rho_{(p,q)}\sum_{x \in
\Sigma_p}{\textbf{P}_p(x)(\ln{\frac{\textbf{P}_p(x)}{\textbf{P}_q(x)}})}.$$

Let us use Lagrange's method to prove the inequalities. Consider the
function $$\phi \doteq \phi(z_{x_1}, z_{x_2}, \dots , z_{x_{|\Sigma_p|}},
\lambda) = \sum_{x \in \Sigma_p}{\textbf{P}_p(x)(\ln{\textbf{P}_p(x)}
- \ln{z_x}) + \lambda z_x} - \lambda.$$ By taking derivatives for
each variable we get that $\phi$ can have minimum only in the
solution points of the system
\begin{equation}
\begin{cases}
\phi_{z_x}' = \lambda - \frac{\textbf{P}_p(x)}{z_x},\\
\sum_{x \in \Sigma_p}{z_x} = 1;
\end{cases}
\end{equation}
and in the boundary points $z_x = 1$ (since $\phi \rightarrow
+\infty$ for $z_x \rightarrow +0$). If for some $x, z_x = 0$ then
for each $y \neq x$ other $z_y = 0 = \textbf{P}_p(y)$ and $\phi =
0$. In the opposite case, we have $\lambda =
\frac{\textbf{P}_p(x)}{z_x}$ for each $x \in \Sigma_p$. Since
$$\sum_{x \in \Sigma_p}{z_x} = \sum_{x \in \Sigma_p}{\textbf{P}_p(x)} = 1,$$
this can happen only if $\textbf{P}_p(x) = z_x$ for each $x \in
\Sigma_p$ whence $\phi = 0$.

Thus $\phi = 0$ if and only if $\textbf{P}_p(x) = z_x$ for each $x
\in \Sigma_p$ such that $\textbf{P}_p(x) > 0$ and $\phi > 0$
otherwise. Due to the condition of states non-equivalence we get
that for some pair $(p,q)$ in $M$ this condition doesn't hold. The
lemma follows.
\end{proof}

\begin{lemma}
\label{lem_lower_bound}For each $M \in \mathrsfs{M}$ and each
$\delta>0$ we have $\mathbb{E}(1/Q_L^{1/L}) \leq
\exp{(\mathbb{E}_{M})} + \delta$ for $L$ big enough.
\end{lemma}
\begin{proof}
First notice that for each pair of states $p_1,p_2$ we get
$$\mathbb{E}(1/Q_L(w \mid p_1)^{1/L}) / \mathbb{E}(1/Q_L(w \mid
p_2)^{1/L}) \rightarrow 1.$$ This follows from the fact that one can
assign to each word $u_1$ generated from $p_1$ the unique word
$v_{p_1,p_2} u_1$ generated from $p_2$, where $v_{p_1,p_2}$ is a
fixed word labeling some path from $p_2$ to $p_1$.
This means that we can take any state $p$ and consider only words
generated from $p$ to estimate the limit of $\mathbb{E}(1/Q_L^{1/L})$.

Given a word $w$ of length $L$ generated from $p$ by
Lemma~\ref{lem_ql_bounds} we get
$$\frac{\textbf{P}_p(w)}{{Q_L}_{p}(w)^{1/L}} \leq \textbf{P}_{p}(w)(\frac{\textbf{P}_{f_w}(w)}{\textbf{P}_{s_w}(w)})^{1/L} .$$
If $f_w \neq p$ then $\textbf{P}_{s_w}(w) \geq \textbf{P}_{p}(w)$.
In the opposite case, $\textbf{P}_{s_w}(w) \geq \textbf{P}_{q}(w)$
for any $q \neq p$. Hence summing up for all length $L$ words we
obtain
\begin{multline}
\mathbb{E}(1/Q_L^{1/L}) \leq \sum_{q \neq p}{\sum_{w \mid f_w = q}{\textbf{P}_{p}(w)(\frac{\textbf{P}_{q}(w)}{\textbf{P}_{p}(w)})^{1/L}}} + \\
+ \sum_{w \mid f_w =
p}{\textbf{P}_{p}(w)(\frac{\textbf{P}_{p}(w)}{\textbf{P}_g(w)})^{1/L}},
\end{multline}
where $g$ is some state such that the pair $(p,g)$ belongs to $M$.
Suppose $(p,q)$ doesn't belong to any $M \in \mathrsfs{M}$. By
Remark~\ref{rem_out_of_M} and inequality~\ref{eq_aep} for both
summands and each $\epsilon>0$ for $L$ big enough we have
\begin{multline}
\mathbb{E}(1/Q_L^{1/L}) \leq (1-\alpha^{L})(\sum_{(p',q') \mid M_{p',q'} }{\frac{\lambda_{p',q'}}{\exp{(\mathbb{E}(M_{p',q'})-\epsilon)}}} + \\
+ \lambda_p \exp{(\mathbb{E}(M)+\epsilon)}) + (c^{0.5L} +
\alpha^{0.5 L})D,
\end{multline}
where $M_{p',q'} \in \mathrsfs{M}$ contains the pair $(p',q')$,
$$\lambda_{p',q'} = \sum_{w \mid f_w = q, (p,q).w[0..0.5L] = (p',q')}{\textbf{P}_{p}(w)} \geq 0, \sum{\lambda_{p',q'}} + \lambda_p = 1,$$ $0 < \alpha < 1$ is from
inequality~\ref{eq_aep} and $D$ is the maximum of
$\frac{1}{\textbf{P}_{q}(x)}$ among all $x \in \Sigma$ and $q \in Q$
such that $\textbf{P}_{q}(x) > 0$. Since
$\exp{(\mathbb{E}(M_{p,q}))}>1$ and for each $\epsilon$ this
inequality holds for $L$ big enough, the lemma follows.
\end{proof}

\begin{lemma}
\label{lem_upper_bound} For $M \in \mathrsfs{M}$ with minimal
$\mathbb{E}_{M}$ and each $\beta>0$ we have $\mathbb{E}(Q_L^{1/L})
\leq 1/\exp{(\mathbb{E}_{M})} + \beta$ for $L$ big enough.
\end{lemma}
\begin{proof}
By Lemma~\ref{lem_ql_bounds} we have $Q_L(w)^{1/L} \leq (c_2
\frac{\textbf{P}_{s_w}(w)}{\textbf{P}_{f_w}(w)})^{1/L}$. Using
Remark~\ref{rem_out_of_M} we also get that $$\mathbb{E}(Q_L^{1/L})
\leq c^{0.5L} + \sum_{(p,q) \in \mathrsfs{M}}{ \lambda_{p,q}
(c_2\frac{\textbf{P}_{q}(w)}{\textbf{P}_{p}(w)})^{1/0.5L} },$$ where
$\sum{\lambda_{p,q}} = 1$.

Thus by inequality~\ref{eq_aep} for each $\epsilon>0$ we get
\begin{multline}
\mathbb{E}(Q_L^{1/L}) \leq c^{0.5L} + \alpha^{0.5L} +
(1-\alpha^{0.5L})\sum_{(p,q) \in \mathrsfs{M}}{ \lambda_{p,q}
(c_2)^{1/0.5L}{\frac{\lambda_{p,q}}{\exp{(\mathbb{E}(M_{p,q})-\epsilon)}}}}.
\end{multline}
for $L$ big enough. The lemma follows.
\end{proof}

Now our second main result follows from Lemma~\ref{lem_lower_bound}
and Lemma~\ref{lem_upper_bound}.
\begin{theorem}
\label{th_main_pred}Let $\mathrsfs{A}$ be a non-exact
$\epsilon$-machine. Then the prediction rate of $\mathrsfs{A}$ is
given by the maximum of $1/\exp{(\mathbb{E}(M))}$ among $M \in
\mathrsfs{M}$. Therefore since $\mathbb{E}(M)>0$ by
Lemma~\ref{lem_about_avg} $pr(\mathrsfs{A}) < 1$ and $\mathrsfs{A}$
 is asymptotically synchronizable.
\end{theorem}

The following corollary follows from the fact that the dominant
eigenvectors of $M  \in \mathrsfs{M}$ can be approximated in
polynomial time.
\begin{corollary}
Given a non-exact $\epsilon$-machine $\mathrsfs{A}$, the prediction
rate of $\mathrsfs{A}$ can be approximated in polynomial time within
any given precision.
\end{corollary}

\section{Conclusions}

Thus it turns out that there are polynomial-time approximation
schemes (PTAS) for computing the natural measures of synchronization
for stochastic setting for both exact and asymptotic cases while no
polynomial time algorithm can approximate the reset threshold even
with logarithmic performance ratio, and even for the binary alphabet
case~\cite{On2Problems} (see also
\cite{MyTOCS2013},\cite{Ep90},\cite{Gerb1}) unless $P=NP$.

The author is grateful to Prof. J\"{u}rgensen for his generous
support during conducting this research.

\end{document}